\documentclass{amsart}

\usepackage[utf8]{inputenc}
\usepackage{amsmath,amsfonts,amssymb,mathtools}
\usepackage{amsthm}
\usepackage[T1]{fontenc}
\usepackage{times}
\usepackage{graphicx}
\usepackage{algorithm}
\usepackage[noend]{algpseudocode}

\algnewcommand{\LineComment}[1]{\State \(\triangleright\) #1}

\newtheorem{theorem}{Theorem}[section]

\theoremstyle{definition}

\newtheorem{definition}[theorem]{Definition}
\newtheorem{remark}[theorem]{Remark}

\newcommand{\F}{{\mathbb F}}
\newcommand{\BigO}[1]{\ensuremath{\operatorname{O}\bigl(#1\bigr)}}
\newcommand{\BigOmega}[1]{\ensuremath{\operatorname{\Omega}\bigl(#1\bigr)}}
\newcommand{\BigTheta}[1]{\ensuremath{\operatorname{\Theta}\bigl(#1\bigr)}}

\begin{document}

\title{On Chudnovsky-Based Arithmetic Algorithms in Finite Fields}

\author{Kevin Atighehchi} 
\address{Kevin Atighehchi, Aix-Marseille Universit{\'e}, Laboratoire d'Informatique Fondamentale de Marseille,
case 901, F13288 Marseille cedex 9 France}
\email{Kevin.Atighehchi@univ-amu.fr}

\author{St\'ephane Ballet} 
\address{St\'ephane Ballet, Aix-Marseille Universit{\'e}, Institut de Math\'{e}matiques
de Marseille case 907, F13288 Marseille cedex 9 France}
\email{Stephane.Ballet@univ-amu.fr}

\author{Alexis Bonnecaze}  
\address{Alexis Bonnecaze, Aix-Marseille Universit{\'e}, Institut de Math\'{e}matiques
de Marseille case 907, F13288 Marseille cedex 9 France}
\email{Alexis.Bonnecaze@univ-amu.fr}

\author{Robert Rolland}
\address{Robert Rolland, Aix-Marseille Universit{\'e}, Institut de Math\'{e}matiques
de Marseille case 907, F13288 Marseille cedex 9 France}
\email{Robert.Rolland@univ-amu.fr}

\begin{abstract}
Thanks to a new construction of the so-called Chudnovsky-Chudnovsky multiplication algorithm, 
we design efficient algorithms for both the exponentiation and the multiplication in finite fields.
They are tailored to hardware implementation and they allow computations to be parallelized 
while maintaining a low number of bilinear multiplications. 
We give an example with the finite field $\F_{16^{13}}$.
\end{abstract}

\maketitle

\section{Introduction}

\subsection{Context}
Multiplication in finite fields is a fundamental operation in arithmetic and finding  efficient 
multiplication methods remains a topical issue.
Let $q$ be a prime power, ${\mathbb F}_q$ the finite field with $q$ elements and ${\mathbb F}_{q^n}$ 
the degree $n$ extension of ${\mathbb F}_q$.
If $\mathcal{B}=\{e_1,...,e_n\}$ is a basis of ${\mathbb F}_{q^n}$ over $\F_q$ then for 
$x=\sum_{i=1}^{n}x_ie_i$ and $y=\sum_{i=1}^{n}y_ie_i$, 
we have the product 
\begin{equation}\label{calculdirect}
z=xy=\sum_{h=1}^{n}z_he_h=\sum_{h=1}^{n}\biggr( \sum_{i,j=1}^{n}t_{ijh}x_ix_j\biggl)e_h,
\end{equation}
 where $$e_ie_j=\sum_{h=1}^{n}t_{ijh}e_h,$$ 
$t_{ijh}\in \F_q$ being some constants.
The complexity of a multiplication algorithm in ${\mathbb F}_{q^n}$ depends on the number of multiplications and additions 
in  $\F_q$. There exist two types of multiplications in $\F_q$: the scalar multiplication 
and the  bilinear  multiplication. The scalar multiplication is the multiplication 
by a constant (in $\F_{q}$) which does not depend on the elements of $\F_{q^n}$ 
that are multiplied. The bilinear multiplication is a multiplication of elements 
that depend on the elements of $\F_{q^n}$ that are multiplied. 
The bilinear complexity is independent 
of the  chosen representation of the finite field.
For example, the direct calculation of $z=(z_1,...,z_n)$ using (\ref{calculdirect}) requires $n^2$ non-scalar multiplication $x_ix_j$, 
$n^3$  scalar multiplications and $n^3-n$ additions.

More precisely, the multiplication of two elements of $\F_{q^n}$ is 
an $\F_q$-bilinear application from $\F_{q^n} \times \F_{q^n}$ onto $\F_{q^n}$.
Then it can be considered as an $\F_q$-linear application from the tensor product 
${\F_{q^n} \otimes_{\F_q} \F_{q^n}}$
onto $\F_{q^n}$. Consequently, it can also be  considered as an element 
$T$ of ${{\F_{q^n}}^\star \otimes_{\F_q} {\F_{q^n}}^\star \otimes_{\F_q} \F_{q^n}}$ where $\star$ denotes the dual.
When $T$ is written
\begin{equation}\label{tensor}
T=\sum_{i=1}^{r} x_i^\star\otimes y_i^\star\otimes c_i,
\end{equation}
where the $r$ elements $x_i^\star$ as well as the $r$ elements $y_i^\star$
are in the dual ${\F_{q^n}}^\star$ of $\F_{q^n}$ while the $r$ elements $c_i$ are in $\F_{q^n}$,
the following holds for any ${x,y \in \F_{q^n}}$:
$$
x\cdot y=\sum_{i=1}^r x_i^\star(x) y_i^\star(y) c_i.
$$
The decomposition (\ref{tensor}) is not unique. 

\begin{definition}
Every expression
$$
x\cdot y=\sum_{i=1}^r x_i^\star(x) y_i^\star(y) c_i
$$
defines a bilinear multiplication algorithm ${\mathcal U}$ of bilinear complexity $\mu({\mathcal U})=r$.
\end{definition}

\begin{definition}
The minimal number of summands in a decomposition of the tensor $T$ of the multiplication
is called the bilinear complexity of the multiplication and is denoted by
$\mu_{q}(n)$:
$$
\mu_{q}(n)= \min_{{\mathcal U}} \mu({\mathcal U})
$$ where ${\mathcal U}$ is running over all bilinear multiplication algorithms in $\F_{q^n}$ over $\F_q$.
\end{definition}

The bilinear complexity of the multiplication in ${\mathbb F}_{q^n}$
over ${\mathbb F}_q$ has been widely studied. 
In particular, it was proved in \cite{ball1} that it is uniformly linear with respect to the degree 
$n$ of the extension. It follows from a clever algorithm performing multiplication: 
the so-called multiplication algorithm of Chudnovsky and Chudnovsky. The original Chudnovsky-Chudnovsky 
algorithm was introduced in 1987 by D.V. and G.V. Chudnovsky \cite{chch} 
and is based on the interpolation on some algebraic curves. From now on, we will denote this algorithm by CCMA.

There is benefit having a low bilinear complexity  when considering hardware 
 implementations mainly because it reduces the  number of gates in the circuit. In fact, in the so-called non-scalar model
  (denoted NS),
 only the bilinear complexity is taken into account and it is
 assumed that all scalar operations are free. Indeed, this model does not reflect the reality and since 
the bilinear complexity is not the whole complexity of the algorithm,  the complexity 
of the linear part of the algorithm should also be taken into account. In this paper, we consider two other models. The
model S1, which takes into account  the number of multiplications without distinguishing between the bilinear ones 
and the scalar ones. The model S2 which takes into account all operations (multiplications and additions) in $\F_q$.

Notice that so far,
 practical implementations of multiplication
 algorithms over finite fields have failed to 
 simultaneously optimize the number of scalar multiplications, additions and bilinear multiplications.

 Regarding exponentiation algorithms, the use of a normal basis is of interest  because the $q^{th}$ power 
 of an element is just a cyclic shift of its coordinates. A remaining question is, how to implement multiplication efficiently 
 in order to have simultaneously fast multiplication and fast exponentiation. 
 In 2000, Gao et al. \cite{GvzGPS2000} show that fast multiplication methods can be
adapted to normal bases constructed with Gauss periods. They show that if $\F_{q^n}$ 
is represented by a normal basis over $\F_q$ generated by a Gauss period of type $(n,k)$, 
the multiplication in $\F_{q^n}$ can be computed with $\BigO{nk\log nk \log\log nk}$ 
and the exponentiation with $\BigO{n^2k\log k \log\log nk}$ operations in $\F_q$ ($q$ being small). This result 
is valuable when $k$ is bounded. However, in the general case $k$ is upper-bounded by $\BigO{n^3\log^2 nq }$.

In 2009, Couveignes and Lercier construct in \cite[Theorem 4]{Cole} 
two families of basis (called elliptic and normal elliptic) for finite field extensions 
from which they obtain a model $\Xi$ defined as follows.
To every couple $(q,n)$, they associate a model, $\Xi(q,n)$, of the degree $n$ extension of $\F_q$ such that the following holds: 

\vspace{.5em}

There is a positive constant $K$ such that the following are true: 

\vspace{.5em}

$\bullet$ Elements in $\F_{q^n}$ are represented by vectors for which the number of components in $\F_q$ is upper bounded by  
$$Kn(\log n)^2\log (\log n)^2.$$

\vspace{.5em}

$\bullet$  There exists an algorithm that multiplies two elements at the expense of $$Kn(\log n)^4|\log (\log n)|^3$$ 
multiplications in $\F_q$.

\vspace{.5em}
 
$\bullet$ Exponentiation by $q$ consists in a circular shift of the coordinates.

\vspace{.5em}

Therefore, for each extension of finite field, they show that there exists a model which allows both fast 
multiplication and fast application of the Frobenius automorphism.
Their model has the advantage of existing for all extensions. 
However, the bilinear complexity of their algorithm is not competitive compared with the best known 
methods, as pointed out in \cite[Section 4.3.4]{Cole}. Indeed, it is clear that such a model requires at least 
$$Kn(\log n)^2(\log (\log n))^2$$  bilinear multiplications.

Note that throughout the paper, efficiency of  algorithms is described in terms of parallel time (depth of the 
circuit, in number of multiplications), number of processors (width) 
and total number of multiplications (size). We have width $\leq$ size $\leq$ depth$.$width. 
\subsection{New results}
We propose another model with the following characteristics:

- Our model is based on CCMA method, thus the multiplication algorithm has a bilinear complexity in $O(n)$, which is optimal. 

- Our model is tailored to parallel computation.  Hence, the computation time used to perform a multiplication or 
any exponentiation can easily be reduced with an adequate number of processors.  
Since our method has a bilinear complexity of multiplication in $O(n)$, it can be parallelized to obtain a constant 
time complexity using $\BigO{n}$ processors. The previous aforementioned works (\cite{GvzGPS2000} and \cite{Cole}) 
do not give any parallel algorithm (such an algorithm is more difficult to conceive than a serial one). 

- Exponentiation by $q$ is a circular shift of the coordinates and can be considered free. Thus, efficient parallelization 
can be done 
when doing exponentiation.

 - The scalar complexity of our exponentiation algorithm is reduced compare to a basic exponentiation using CCMA algorithm 
thanks to a suitable basis representation of the Riemann-Roch space ${\mathcal L}(2D)$ in the second evaluation map.  
More precisely, the normal basis
representation of the residue class field is carried in the associated Riemann-roch space ${\mathcal L}
(D)$, and the exponentiation by $q$ consists in a circular 
shift of the $n$ first coordinates of the vectors lying in the Riemann-Roch space ${\mathcal L}(2D)$.

 - Our model uses Coppersmith-Winograd \cite{CW1987} method (denoted CW) or any variants 
 thereof to improve matrix products  and to diminish the number of scalar operations. 
 This improvement is particularly efficient for exponentiation.

In term of complexity, we can state the following results, depending on the chosen model
(NS, S1 and S2).

\begin{theorem}\label{LNS}
In the non-scalar model NS,  there exist multiplication and exponentiation algorithms  in $\F_{q^n}$  such that: 
\begin{itemize}
\item[-] Multiplication is done in parallel time in $\BigO{1}$ multiplications in $\F_q$ with $\BigO{n}$ processors, 
for a total in $\BigO{n}$ multiplications.
\item[-]  Exponentiation is done in parallel time in $\BigO{\log n}$ multiplications in $\F_q$ with $\BigO{n^2/\log^2n}$ processors, 
for a total in $\BigO{n^2/\log n}$ multiplications.
\end{itemize}
\end{theorem}
When considering models S1 and S2, two cases can be distinguished for the multiplication complexity. 
We might be interested either by the
complexity of one multiplication or by the average (amortized) complexity of one multiplication when 
many multiplications are done simultaneously.
Regarding exponentiation, a wise use of CW method allows the complexity to be improved. 

We can state the followings:
\begin{theorem}\label{LS1}
In the model S1, there exist multiplication and exponentiation algorithms in $\F_{q^n}$ such that:
\begin{itemize}
\item[-] multiplication:
   \begin{itemize}
      \item[a)] one multiplication is done in parallel time in $\BigO{1}$ multiplications in $\F_q$ with $\BigO{n^2}$ processors, 
for a total in $\BigO{n^2}$ multiplications;
      \item[b)] in the amortized sense, the parallel time is in $\BigO{1}$ multiplications 
      in $\F_q$ with $\BigO{n^{1+\epsilon}}$ processors, 
for a total in $\BigO{n^{1+\epsilon}}$ multiplications where the value of $\epsilon$ is approximately $0.38$ for 
the best known matrix product methods;
   \end{itemize}
\item[-]  exponentiation is done in a parallel time of 
$\BigO{\log n}$ multiplications in $\F_q$ with $\BigO{n^{2+\epsilon}/\log^{2\epsilon} n}$ processors, 
for a total in $\BigO{n^{2+\epsilon}\log^{1-2\epsilon} n}$ multiplications.
\end{itemize}
\end{theorem}
\begin{theorem}\label{LS2}
In the model S2, there exist multiplication and exponentiation algorithms in $\F_{q^n}$ such that:
\begin{itemize}
\item[-] multiplication: 
  \begin{itemize}
    \item[a)] one multiplication is done in parallel 
time in $\BigO{\log n}$ operations in $\F_q$ with $\BigO{n^2/\log n}$ processors, 
for a total in $\BigO{n^2}$ operations;
    \item[b)] in the amortized sense, the parallel time is in $\BigO{\log n}$ operations 
    in $\F_q$ with $\BigO{n^{1+\epsilon}/\log n}$ processors, 
for a total in $\BigO{n^{1+\epsilon}}$ operations; recall that the value of $\epsilon$ 
is approximately $0.38$ for the best matrix product methods;
   \end{itemize}
\item[-]  exponentiation is done in a parallel time of $\BigO{\log^2 n}$ 
operations in $\F_q$ with $\BigO{n^{2+\epsilon}/\log^{1+2\epsilon} n}$ processors, 
for a total in $\BigO{n^{2+\epsilon}\log^{1-2\epsilon} n}$ operations.

\end{itemize}
\end{theorem}

\subsection{Organization of the article} 
 After some background on  CCMA algorithm, we describe 
in Subsection~\ref{construction} our method which leads 
to  an effective algorithm  that can directly be implemented. Our algorithm reveals the use of 
matrix-vector products that can easily be parallelized.
In Section~\ref{secExp}, we use 
this algorithm  
to  tackle the problem of computing $x^k$ where $x \in {\mathbb F}_{q^n}$ and $k \geq 1$
and we derive an exponentiation algorithm from the work of von zur Gathen \cite{Gat91,Gat92exp}.
In Section~\ref{effective}, 
we focus on the multiplication in $\F_{16^{n}/\F_{16}}$ and we explain how to construct our
algorithm. A Magma~\cite{Magma} implementation of the multiplication 
algorithm in $\F_{16^{13}/\F_{16}}$ is given in appendix.

\section{A new approach of multiplication and exponentiation algorithms}

First, we present the CCMA algorithm on which is based our method.

\subsection{Original algorithm of Chudnovsky-Chudnovsky (CCMA)}\label{Chud}

Let $F/{\mathbb F}_q$ be an algebraic function field over the finite field ${\mathbb F}_q$
of genus $g(F)$. We denote by $N_1(F/{\mathbb F}_q)$ the number of places of degree one of $F$ over
${\mathbb F}_q$. If $D$ is a divisor, ${\mathcal L}(D)$ denotes the Riemann-Roch space
associated to $D$. We denote by ${\mathcal O}_Q$ the valuation ring of the place $Q$ and
by $F_Q $ its residue class field ${\mathcal O}_Q/Q$
which is isomorphic to $\F_{q^{{\rm deg} (Q)}}$ where ${\rm deg} (Q)$ is the degree of the place $Q$.
The following theorem that makes effective the original algorithm groups some results of \cite{ball1}.

\begin{theorem}\label{mainth}
Let $F/\F_q$ be an algebraic function field of genus $g(F)$ defined over $\F_q$ and
$n$ an integer.
Let us suppose that there exists a place  $Q$ of degree $n$. 

Then, if $N_1(F/{\mathbb F}_q)>2n+2g-2$ 
there is an effective divisor $D$ of degree $n+g-1$ such that:
\begin{enumerate}
\item $Q$ is not in the support of $D$,
\item the evaluation map $E$ defined by
$$\begin{array}{lccc}
 E: & {\mathcal L}(D) & \rightarrow & F_Q \\
    &  f              & \mapsto     & f(Q)
\end{array}$$
is an isomorphism of vector spaces over ${\mathbb F}_q$,
\item there exist $2n+g-1$  
places of degree one $P_i$ which are not in the support of $D$
such that the multi-evaluation map $T$ defined by
$$\begin{array}{lccl}
 T: & {\mathcal L}(2D) & \rightarrow & \left({\mathbb F}_q\right)^{2n+g-1} \\
    &  f              & \mapsto     & \left(\strut f\left(P_1\right),\ldots,f\left(P_{2n+g-1}\right)\right)
\end{array}$$
is an isomorphism.
\end{enumerate}
\end{theorem}

The chosen framework is the original CCMA algorithm, 
namely using only places of degree one and without derivated evaluation (cf. \cite{ceoz}).
We transform this algorithm in order that it be adapted to both  multiplication and  exponentiation computations.


In this context, the construction of this algorithm is based on the choice of the place $Q$ of degree $n$,  
the effective divisor $D$ (cf. \cite{ball2}) and the bases $ {\mathcal L}(D) $ and $ {\mathcal L}(2D) $.

\subsection{Normal bases}

Recall some notions on normal bases. The finite field ${\mathbb F}_{q^n}$ will be considered as
a vector space of dimension $n$ over the finite field ${\mathbb F}_q$.
Let $\alpha$ be an element of ${\mathbb F}_{q^n}$ such that 
$$\left(\alpha, \alpha^q, \alpha^{q^2}, \ldots, \alpha^{q^{n-1}}\right)$$
is a basis of ${\mathbb F}_{q^n}$ over ${\mathbb F}_q$. Such a basis is called
a normal basis of ${\mathbb F}_{q^n}$ over ${\mathbb F}_q$ and $\alpha$ is called a cyclic element.
Thus, a normal basis is composed of all conjugates of a cyclic element $\alpha$. 
There is  always  a normal basis and furthermore, there is  always a primitive normal basis.
We call a normal polynomial of degree $n$ over ${\mathbb F}_q$,  a polynomial in ${\mathbb F}_q[X]$, 
irreducible over ${\mathbb F}_q$, and 
having for roots
in ${\mathbb F}_{q^n}$ the $n$ conjugates of a cyclic element $\alpha$.
We refer to \cite{lini2} and \cite{gao} for a detailed presentation. 

\medskip
When $\F_{q^n}$ is represented by a normal basis, the $q$th power of an element is just a cyclic shift of 
its coordinates. The repeated use of this operation allows exponentiation to be efficiently parallelized.
Without normal basis \cite{LimL94}, precomputation should be stored for a same base $x$. This makes sense only when
many exponentiation have to be done with this same base and in this case, precomputations are not 
considered in the running time.  

The use of a normal basis has the following benefits:
\begin{itemize}
 \item Substitute lookup table accesses by  circular shifts.
  \item Reduce prior storage.
 \item Avoid the constraint of fixing a base. 
\end{itemize}

\subsection{Method and strategy of implementation}\label{construction}
The construction of the algorithm is based on the choice of the place $Q$ of degree $n$,  
the effective divisor $D$ of degree $n+g-1$ (cf. \cite{ball2}), the bases of spaces ${\mathcal L}(D) $
 and ${\mathcal L}(2D)$ and 
the basis of the residue class field $F_Q$  of the place $Q$. 

In practice, as in \cite{ball2}, we take as a divisor $D$ one place of degree $n+g-1$.
It has the advantage to solve the problem of the support of divisor $D$ (condition (1) of Theorem 
\ref{mainth}) as well as the problem of the effectivity of the divisor $D$.
Furthermore, we require additional properties.

\subsection{Finding places $D$ and $Q$}\label{si}

To build the good places, we draw them at random and  we check that they satisfy the required 
conditions namely :

\begin{enumerate}
\item We draw at random an irreducible polynomial $\mathcal{Q}(x)$ of degree $n$ in $\F_q[X]$ and check that this polynomial is :
     \begin{enumerate}
           \item Primitive.
           \item Normal.
           \item Totally decomposed in the algebraic function field $F/\F_q$ (which implies that there exists a place $Q$ of degree n  
           above the polynomial $\mathcal{Q}(x)$).           
     \end{enumerate}
     \item We choose a place $Q$ of degree n among the $n$ places lying above the polynomial $\mathcal{Q}(x)$. 
     \item We draw at random a place $D$ of degree $n+g-1$ and check that $D-Q$ is a non-special divisor of degree $g-1$, 
     {\it i.e.} $dim {\mathcal L}(D-Q)=0$.
\end{enumerate}

\begin{remark}
In practice, it is easy to find $Q$ and $D$ satisfying these properties in our context since there exist many such places. 
However, it is not true in the  general case (it is sufficient to consider an elliptic curve with only one rational point). 
A sufficient condition for the existence of at least one place of degree $n$ is given by the following inequality:
$$2g(F)+1 \leq q^{\frac{n-1}{2}}\left(q^{\frac{1}{2}}-1\right).$$

When $q\geq 4$, we are sure of the existence of a non-special divisor of degree $g-1$ \cite{balb}.
The larger $q$, the larger  the probability to draw a non-special divisor of degree $g-1$ becomes (Proposition 5.1 \cite{bariro})
but not necessarily as a difference of two places (this is an open problem). However, in practice, 
such divisors are easily found.

\end{remark}
\subsection{Choice of bases of spaces} \label{choixbases}

\subsubsection{The residue field $F_Q$}\label{baseFQ}

When we take a place $Q$ of degree $n$ lying above a normal polynomial in $\F_q[X]$, we mean that the residue class field 
is the finite field $\F_{q^n}$ for which we choose as a representation basis  the normal 
basis $\mathcal B_Q$ generated by a root $\alpha$ of the 
polynomial $\mathcal{Q}(x)$.

\begin{remark}\label{changebasis}
Suppose that the context requires the use of a representation basis of $\F_{q^n}$ which is not the basis $\mathcal{B}_Q$
of the residue class field of the place $Q$. Then, we can easily avoid the problem by a change of basis. 
This requirement may happen  when additional properties on the basis $\mathcal{B}_Q$ 
(cf. Section \ref{ChudNB}) are required. In particular, it would be the case in our context if we could not find a place $Q$ 
of degree $n$ 
above a normal polynomial $\mathcal{Q}(x)\in \F_q[X]$.
\end{remark}

\subsubsection{The Riemann-Roch space $\mathcal{L}(D)$}\label{baseLD}

As the residue class field $F_Q$ of the place $Q$ is isomorphic to the finite field $\F_{q^n}$, 
from now on we identify $\F_{q^n}$ to $F_Q$. 
Notice that the choice of $D$ and $Q$ of Section \ref{si} are such that the map $E$ of Theorem~\ref{mainth} is an isomorphism. 
Indeed, $\deg (D)=n+g-1$,  $dim (D-Q)=0$ yet ${\mathcal L}(D-Q)=Ker(E).$
In particular, we choose for basis of ${\mathcal L}(D)$, the reciprocal 
image $\mathcal{B}_D$ of the basis $\mathcal{B}_Q=(\phi_1,\ldots,\phi_n)$ of  $F_Q$ 
by the evaluation map $E$, namely $\mathcal{B}_D=(E^{-1}(\phi_1),\ldots,E^{-1}(\phi_n))$. 

\subsubsection{The Riemann-Roch space $\mathcal{L}(2D)$} \label{baseL2D}
Note that as the divisor $D$ is an effective divisor, we have ${\mathcal L}(D)\subset {\mathcal L}(2D)$. 
Let $P$ be the map from ${\mathcal L}(2D)$ to ${\mathcal L}(2D)$
defined in the following way:
if $f \in {\mathcal L}(2D)$ then $f(Q)$ is in the residue field $F_Q$ of the place $Q$;
define $P(f) =J\circ E^{-1}\left(\strut f(Q)\right)$ where $J$ is the injection
map from  ${\mathcal L}(D)$ into ${\mathcal L}(2D)$. Then $P$
is a linear map from ${\mathcal L}(2D)$ into ${\mathcal L}(2D)$ whose image is  ${\mathcal L}(D)$.
More precisely,
$P$ is a projection from ${\mathcal L}(2D)$ onto ${\mathcal L}(D)$.
Let ${\mathcal M}$ be the kernel of $P$. Then
$${\mathcal L}(2D)={\mathcal L}(D) \oplus {\mathcal M}.$$

\begin{remark}\label{rem_proj} From the definition of $P$ we remark that
\begin{enumerate}
\item ${\mathcal M}=\{ f\in {\mathcal L}(2D)~|~f(Q)=0\}$,
\item for any $f \in {\mathcal L}(2D)$, we have $f(Q)=P(f)(Q)$.
\end{enumerate}
\end{remark}


As $\deg 2D > 2g-2$, the divisor $2D$ is non-special and $$\dim {\mathcal L}(2D)=2n+g-1.$$
Hence, we define as basis of ${\mathcal L}(2D)$, the basis $\mathcal{B}_{2D}$ defined by: 

$$\mathcal{B}_{2D}=(f_1, \ldots ,f_n,f_{n+1},\ldots , f_{2n+g-1})$$
where $\mathcal{B}_{{\mathcal M}}=(f_{n+1}, \ldots , f_{2n+g-1})$ is a basis  of ${\mathcal M}$ 
and $\mathcal{B}_{D}=(f_1, \ldots ,f_n)$ is the basis of ${\mathcal L}(D)$ defined in Section \ref{baseLD}.

\begin{remark}\label{mainrem}
As a consequence of the choice of the basis
$$\mathcal{B}_{2D}=(f_1, \ldots ,f_n,f_{n+1}, \ldots , f_{2n+g-1})$$
for ${\mathcal L}(2D)$, if  
$$x=(x_1,\ldots,x_n,x_{n+1},\ldots,x_{2n+g-1})\in {\mathcal L}(2D)$$
then 
$$P(x)=(x_1,\ldots,x_n,0,\ldots,0).$$
\end{remark}

\subsection{Product of two elements in ${\mathbb F}_{q^n}$}

In this section, we use as representation basis of spaces $F_Q$, ${\mathcal L}(D)$, ${\mathcal L}(2D)$, 
the basis defined in Section \ref{choixbases}. The product of two elements in ${\mathbb F}_{q^n}$ 
is computed by the algorithm of Chudnovsky and Chudnovsky.
Let $x=(x_1,\ldots,x_n)$ and $y=(y_1,\ldots,y_n)$ be 
two elements of ${\mathbb F}_{q^n}$ given by their components over ${\mathbb F}_{q}$
relative to the chosen basis $\mathcal{B}_Q$. 
According to the previous notation, we can consider that $x$ and $y$ are 
identified to the following elements of ${\mathcal L}(D)$:
$$f_x=\sum_{i=1}^n x_if_i \quad \hbox{and} \quad f_y=\sum_{i=1}^n y_if_i.$$
The product $f_xf_y$ of the two elements $f_x$ and $f_y$ of ${\mathcal L}(D)$ is their product in 
the valuation ring ${\mathcal O}_Q$. This product lies in ${\mathcal L}(2D)$.
We will consider that $x$ and $y$ are respectively the elements
$f_x$ and $f_y$ of ${\mathcal L}(2D)$
where the $n+g-1$ last components are $0$. Now it is clear that knowing $x$ or
$f_x$ by their coordinates is the same thing.
Let us consider the following Hadamard product in $\left({\mathbb F}_q\right)^{2n+g-1}$:
$$\begin{array}{l}(u_1,\ldots,u_{2n+g-1}) \odot (v_1,\ldots,v_{2n+g-1})\\= 
(u_1v_1,\ldots,u_{2n+g-1}v_{2n+g-1}).
\end{array}
$$

\begin{theorem}
The product of $x$ by $y$ is such that
$$f_{xy}=P\left(T^{-1} \left(\strut T(f_x)\odot T(f_y)\right)\right).$$
\end{theorem}
\begin{proof}
Indeed, from the definition of $T$ the following holds
$$T(f_x) \odot T(f_y) = T(f_{x}f_{y}).$$
Then 
$$P\left(T^{-1} \left(\strut T(f_x)\odot T(f_y)\right)\right)=
P(f_xf_y).$$
By Remark \ref{rem_proj} we conclude 
$$P(f_xf_y)=f_{xy}.$$
\end{proof}

We can now present the setup algorithm and the multiplication algorithm.
Note that the setup algorithm is only done once.

\begin{algorithm}[H]
\caption{Setup algorithm}
\begin{algorithmic}
\Require $F/{\mathbb F}_{q}, ~  Q,  D, P_1,\ldots, P_{2n+g-1}$.
\Ensure  $T \hbox{ and } T^{-1}.$
        \begin{enumerate}
          \item The elements $x$ of the field
          ${\mathbb F}_{q^n}$ are known by their components relatively to a
fixed basis: $x=(x_1,\ldots,x_n)$ (where $x_i \in {\mathbb F}_{q}$).
          \item The function field $F/{\mathbb F}_{q}$, the place $Q$, the
divisor $D$ and the points $P_1,\ldots, P_{2n+g-1}$ are as in Theorem
\ref{mainth}.
          \item Construct a basis
$(f_1, \ldots, f_n,f_{n+1}, \ldots, f_{2n+g-1})$ of ${\mathcal L}(2D)$ where
$(f_1, \ldots, f_n)$ is the basis of ${\mathcal L}(D)$ defined in section \ref{baseLD} and
$(f_{n+1}, \ldots,  f_{2n+g-1})$ a basis of ${\mathcal M}$.
          \item Any element $x=(x_1, \ldots, x_n)$ in ${\mathbb F}_{q^n}$ is
identified to the element $f_x=\sum_{i=1}^n x_if_i$ of ${\mathcal L}(D)$.
          \item Compute the matrices $T$ and $T^{-1}$.
        \end{enumerate}
\end{algorithmic}
\end{algorithm}

\begin{algorithm}[H]
\caption{Multiplication algorithm}
\begin{algorithmic}
\Require $x=(x_1,\ldots,x_n) \hbox{ and } y=(y_1,\ldots,y_n)$.
\Ensure  $xy.$
       \begin{enumerate}
         \item Compute 
$$\left (\begin{array}{c} z_1\\ \vdots\\ z_n\\z_{n+1} \\ \vdots\\ z_{2n+g-1}
\end{array} \right)= T\left (\begin{array}{c} x_1\\ \vdots\\ x_n\\ 0\\ \vdots\\
0 \end{array} \right)  \hbox{ and } 
\left (\begin{array}{c} t_1\\ \vdots\\t_n\\t_{n+1}\\ \vdots\\ t_{2n+g-1}
\end{array} \right)= T\left (\begin{array}{c} y_1\\ \vdots\\ y_n\\ 0\\ \vdots\\
0 \end{array} \right).
$$ 
        \item Compute $u=(u_1,\ldots, u_{2n+g-1})$ where $u_i=z_it_i$.
        \item Compute $w=(w_1,\ldots,w_{2n+g-1})=T^{-1}(u)$.
        \item Return$(xy=(w_1,\ldots,w_n))$ (remark that in the previous step
we just have  to compute the $n$ first components of $w$).
       \end{enumerate}
\end{algorithmic}
\end{algorithm}

In terms of number of multiplications in $\F_q$, the complexity of this multiplication algorithm is as follows:
calculation  of $z$ and $t$ needs $2(2n^2+ng-n)$ multiplications,
calculation of $u$ needs $2n+g-1$ multiplications and calculation of $w$ needs 
$2n^2+ng-n$ multiplications.  
The total complexity  is bounded by $6n^2+n(3g-1)+g-1$.

The asymptotic analysis of our method needs to consider infinite families of algebraic 
function fields defined over $\F_q$ with increasing genus 
(or equivalently of algebraic curves) having the required properties.
The existence of such families follows from that of recursive towers of algebraic function fields of 
type Garcia-Stichtenoth \cite{gast} reaching the Drinfeld-Vladut bound. 

\begin{remark} Note that,  because of the Drinfeld-Vladut bound, in the case of small basis fields $\mathbb{F}_2$ respectively 
$\mathbb{F}_q$ with $3\leq q<9$, we select the method of the quartic respectively quadratic embedding. 
In these cases, instead of the embeddings, it would be possible to  use places of degree four and respectively two 
(cf. \cite{baro1} and \cite{bapi}) but it requires to generalize our algorithm and it generates significant complications (cf. \cite{tuku}, \cite{babotu}).  
Note also that, in the case of the quartic respectively quadratic embedding of the small fields, the operations are precomputed and  
do not increase the complexity.
\end{remark}

Then it is proved \cite{ball1} from a specialization of the original 
 Chudnovsky algorithm on these families that the bilinear complexity of the multiplication in any degree $n$ extension 
  of $\F_q$ is uniformly linear in $q$ with respect to $n$. Hence, the number of bilinear 
  multiplications is in $O(n)$ and the genus $g$ of the required curves also necessarily increases in $O(n)$.
Consequently, the total number of multiplications/additions/subtractions of our method 
is in $\BigO{n^2}$ and the total number of bilinear multiplications is in $\BigO{n}$. 
 
On some occasions in the paper, the CW algorithm \cite{CW1987} will be used to decrease the number of scalar operations.
Given two square matrices of size $n$, the product can be computed in $\BigO{n^{2+\epsilon}}$ multiplications where $\epsilon<2.38$. 
In fact, if we consider a parallel version of the algorithm in the model S1, a product can be performed in $\BigO{1}$ multiplications
in $\F_q$ using $\BigO{n^{2+\epsilon}}$ processors. This is a consequence of Strassen's normal 
form theorem \cite{Str91}. In the model S2 where scalar multiplications and additions/subtractions have same costs, 
the depth becomes $\BigO{\log n}$ and a rescheduling technique \cite{LD90} allows the reduction of 
the width in $\BigO{n^{2+\epsilon}/\log n}$.

Now, we focus on the parallel complexity of our multiplication algorithm. This actually consists of determining
the parallel complexity of a constant number of matrix-vector products and the product coordinate-wise of two vectors.
First, let us  consider the NS model in which a round represents the time interval for a processor 
to perform one bilinear multiplication in $\mathbb{F}_q$ (scalar operations are considered as free), 
a multiplication can be carry out in a constant number of rounds  with only $\BigO{n}$ processors, as stated in Theorem~\ref{LNS}. 

Regarding the two other models, two cases have to be considered, the non amortized case and the amortized one.
In the model S1 (additions and subtractions in $\mathbb{F}_q$ are considered as free), if a round represents the 
time interval for a processor to perform one multiplication, 
the product can be performed in a constant number of rounds 
with $\BigO{n^2}$ processors, as stated in Theorem~\ref{LS1}a. 
This width corresponds to the asymptotic number of scalar multiplications.
In the model S2 where we take into account all scalar operations in $\mathbb{F}_q$, a 
multiplication in $\mathbb{F}_{q^n}$ can be performed in parallel time $\BigO{\log n}$ operations in
$\mathbb{F}_q$ using $\BigO{n^2/\log n}$ processors, as stated in Theorem~\ref{LS2}a.

If we have $\BigOmega{n}$ multiplications in $\mathbb{F}_{q^n}$ to perform,  the width 
 can be decreased (in an amortized sense)  by using
an optimized matrix multiplication method. It consists in grouping the operands in order that they 
become the columns of square matrices allowing the use of  an efficient method like the CW one. 
More precisely, the method is as follows:

\begin{enumerate}
 \item Store the $a=\BigOmega{2n}$ operands (vectors) as columns in square matrices of size $2n+g-1$, 
 denoted $(B_i)_{i=1 \ldots \lceil a/(2n+g-1) \rceil}$;
 \item Compute the products 
 $$(TB_i)_{i=1 \ldots \lceil a/(2n+g-1) \rceil}=(M_{i,1},M_{i,2},M_{i,3},M_{i,4},M_{i,5},\ldots), $$ 
 where $M_{i,j}$ represent the columns of $TB_i$, by using the Coppersmith-Winograd algorithm (or another efficient method);
 \item 
 Perform the $a/2$ bilinear products $(M_{i,1},M_{i,2}),~(M_{i,3},M_{i,4}), \ldots$;
 \item Store the $a/2$ results as columns in square matrices, denoted \break $(B'_i)_{i=1 \ldots \lceil a/(4n+2g-2) \rceil}$;
 \item Compute the products $(T^{-1}B'_i)_{i=1 \ldots \lceil a/(4n+2g-2) \rceil}$ with the CW method. 
 The results are then stored in the columns of the resulting matrices.
\end{enumerate}

In the model S1, the necessary width to perform step 3 is in $\BigO{a}$ processors. 
The overall width is in fact dominated by steps 2 and 5.
All these products are performed in constant time using $\BigO{n^{2.38}+an^{1.38}}$ processors. 
This represents, in an amortized sense, $\BigO{n^{1.38}}$ processors
per multiplication, as stated in Theorem~\ref{LS1}b. 

In the model S2, a multiplication in $\mathbb{F}_{q^n}$ can be performed in parallel time $\BigO{\log n}$ 
using (thanks to the CW matrix product) $\BigO{n^{1.38}/\log n}$ processors, as stated in Theorem~\ref{LS2}b. 
This last result is obtained using
a rescheduling technique \cite{LD90} which allows the parallel computation of the matrix product to be 
done in $\BigO{\log n}$ operations in $\mathbb{F}_q$
using $\BigO{n^{2.38}/\log n}$ processors, instead of $\BigO{n^{2.38}}$ processors without rescheduling \cite{CW90,GM88}.

\subsection{Product of three elements in ${\mathbb F}_{q^n}$}
The previous algorithm can be iterated in order to obtain the product of three elements $x,y,z$ 
in ${\mathcal L}(D)$ (or equivalently in ${\mathbb F}_{q^n}$).

\begin{algorithm}[H]
\caption{Product of three elements }
\begin{algorithmic}
\Require $x,~y,~z \in {\mathbb F}_{q^n}$.
\Ensure  $xyz.$
\begin{enumerate} 
 \item Compute $$u=T\circ P \circ T^{-1} \left(\strut T(x)\odot T(y)\right),$$
 \item compute $v=T(z)$,
 \item then compute $w=u \odot v$ (this is the Hadamard product in $\left({\mathbb F}_q\right)^{2n+g-1}$),
 \item and finally the result is $P\circ T^{-1}(w)$.
\end{enumerate}

\end{algorithmic}
\end{algorithm}

In terms of number of multiplications in $\F_q$, the complexity of this algorithm is as follows:
the matrix $T_1:=T\circ P \circ T^{-1}$ can be precomputed and the matrix-vector product needs $(2n+g-1)^2$ multiplications.
The total complexity is then $12n^2+n(8g-4)+g^2-1$ including $2(2n+g-1)$ 
bilinear multiplications (this number is almost doubled compare to the preceding algorithm). 
Note that the precomputation of $T_1$ is of interest for the parallel computations.
Asymptotically, the complexity is the same as in the previous case.

This algorithm will be used to construct our exponentiation algorithms in Section~\ref{secExp}. The form of 
the involved matrices is analysed in the next subsection.
\subsection{Form of the involved matrices}
Obviously, the crucial part of the algorithm consuming more time is
the iterative call to $T_1$, namely the iterative call to the composition $T\circ P \circ T^{-1}$.
The matrix $T$  is 
$$
\left(
\begin{array}{ccc}
f_1(P_{1}) 
& \ldots & f_{2n+g-1}(P_{1}))\\
f_1(P_{2}) 
& \ldots & f_{2n+g-1}(P_{2}))\\
\vdots     
& \vdots &    \vdots\\
f_1(P_{n}) 
& \ldots & f_{2n+g-1}(P_{n}))\\
\vdots     
& \vdots &    \vdots\\
f_1(P_{2n+g-1}) 
& \ldots & f_{2n+g-1}(P_{2n+g-1}))
\end{array}
\right).
$$
The matrix $P$ is
$$
\left(
\begin{array}{ccc|ccc}
   &  &  &  &  & \\
   &  &  &  &  & \\ 
   & I_n & &  & 0 & \\
   &  &  &  &  & \\
   &  &  &  &  & \\  
\hline
  &  &  &  &  & \\
    & 0 & &  & 0 & \\
  &  &  &  &  & \\
\end{array}
\right)
$$
where $I_n$ is the unit matrix of size $n\times n$.

Let us define the following blocks:
$$
T=\left(
\begin{array}{ccc|ccc}
   &  &  &  &  & \\
   &  &  &  &  & \\ 
   & U_1 & &  & U_3 & \\
   &  &  &  &  & \\
   &  &  &  &  & \\  
\hline
  &  &  &  &  & \\
    & U_2 & &  & U_4 & \\
  &  &  &  &  & \\
\end{array}
\right)
$$
where $U_1$ is a  $n\times n$-matrix, $U_2$ is a  $(n+g-1)\times n$-matrix, 
$U_3$ is a  $n \times (n+g-1)$-matrix
and $U_4$ is a $(n+g-1)\times (n+g-1)$-matrix.
In the same way, we can write $T^{-1}$ in the following form:
$$
T^{-1}=\left(
\begin{array}{ccc|ccc}
   &  &  &  &  & \\
   &  &  &  &  & \\ 
   & V_1 & &  & V_3 & \\
   &  &  &  &  & \\
   &  &  &  &  & \\  
\hline
  &  &  &  &  & \\
    & V_2 & &  & V_4 & \\
  &  &  &  &  & \\
\end{array}
\right).
$$

Then
$$
T_1=\left(
\begin{array}{ccc|ccc}
   &  &  &  &  & \\
   &  &  &  &  & \\ 
   & U_1V_1 & &  & U_1V_3 & \\
   &  &  &  &  & \\
   &  &  &  &  & \\  
\hline
  &  &  &  &  & \\
    & U_2V_1 & &  & U_2V_3 & \\
  &  &  &  &  & \\
\end{array}
\right).
$$
Moreover, the following relations hold:
$$U_1V_1+U_3V_2=I_n,$$
$$U_1V_3+U_3V_4=0,$$
$$U_2V_1+U_4V_2=0,$$
$$U_2V_3+U_4V_4=I_{n+g-1}.$$

An important problem, which is out of the scope of this paper, is to choose a 
basis $(f_i)$ of ${\mathcal L}(2D)$ and places of degree one $P_i$
in order to obtain a ``simple'' matrix $T_1$. Indeed, a sparse matrix may reduce the number of multiplications.

\subsection{Precomputation and storage of scalar multiplications}
Having a matrix $T$ (or $T_1$), the product $Tx$ for all possible $x=(x_1,\ldots , x_n,0,\ldots, 0)\in \mathcal{L}(2D)$
can be efficiently precomputed. We choose an integer $l$ dividing $n$ and we set $k=n/l$.
For all $i$ such that $0 <i<k-1$,
we denote by $X_{l,i}$ a variable of the form
$(X_1,X_2, \ldots, X_n,0,0, \ldots,0) \in (\mathbb{F}_q)^{2n+g-1}$,
where:
$$X_j= \left\{ 
\begin{array}{lll}
0 & \hbox{if} & 1 \leq j \leq i\cdot l \\
x_j & \hbox{if} & i \cdot l +1 \leq j \leq (i+1)l \\
0 & \hbox{if} & (i+1)l+1 \leq j \leq n.
\end{array}
\right .
$$

Thus, we have $x=X_{l,0} + X_{l,1} + \cdots + X_{l,k-1}$.
We can store in a table $Tab^T_i[\cdot]$ the products $TX_{l,i}$ for the $q^l$ possible values of $X_{l,i}$.
After having stored all the precomputations in tables $(Tab^T_i[\cdot])_{i=0 \ldots k-1}$, we can compute
$Ty$ by evaluating $Tab^T_0[Y_{l,0}] + Tab^T_1[Y_{l,1}] + \cdots + Tab^T_{k-1}[Y_{l,k-1}]$.

As an example, for $\mathbb{F}_{16^{13}}$ we can choose $l=2$, leading to the storage of $6$ tables of $256$ values
plus one table of $128$ values.

\section{Exponentiation algorithms}\label{secExp}
In this section, our aim is to point out the interest, in terms of parallel running time and  number of 
 processors involved, of computing an exponentiation in finite fields ({\it i.e.} $x^k$ in ${\mathbb F}_{q^n}$ 
 with $k\in\mathbb N^*$) based on our model. 
 First, it is natural to consider the well known square and multiply algorithm  (with say, the method ``right-to-left'').
 We describe  this basic algorithm, showing the use of matrices $T$, $P$ and $T_1$.
In a second time,  we consider a more advanced algorithm, based on an idea from von zur Gathen \cite{Gat91} 
that we embed in our model. We show that this algorithm  
reaches optimal depth in terms of operations in $\F_q$.

\subsection{Exponentiation algorithm with a square and multiply method}\label{SQM}
Let $K$ be
the unidimensional array of length $s+1$ containing  
the bits of $k$. This array will be indexed by $i=0,\ldots,s$. More precisely
$$k=\sum_{i=0}^s K[i]2^i.$$

In order to bind operations we must iterate the use of the operator $T_1$.
We obtain Algorithm \ref{SaM_algo}:
\begin{algorithm}
\caption{Square-and-Multiply algorithm (right-to-left)}
\label{SaM_algo}
\begin{algorithmic}
\Require $x,~k$
\Ensure $x^k$
\State $X_0 \leftarrow T(x)$
\State $X \leftarrow (1, 1,\ldots, 1)\in \{0,1\}^{2n+g-1}$
\For{$i$ \textbf{from} $0$ \textbf{to} $s$} 
\If{$K[i]==1$}
  \State $X \leftarrow T_1(X \odot X_0)$
\EndIf
\State $X_0 \leftarrow T_1(X_0 \odot X_0)$
\EndFor
\State $Y \leftarrow P\circ T^{-1}(X)$
\State \textbf{return} $Y$
\end{algorithmic}
\end{algorithm}

Within a same loop turn, operations under the condition "if" are not used in the 
subsequent operations. We consider two sets of processing units $P_1$ and $P_2$,
the set $P_2$  running the operations under the condition "if" while the set 
$P_1$ deals with the other calculations. We assume that the amounts of resources of 
$P_1$ and $P_2$ are the same. 
As an example, let us describe the steps in the calculation of $x^{15}$.
Figure~\ref{exponentiation_wo_precomput} depicts the operations made in parallel.
We remark that at each step, the set $P_1$ or $P_2$ (or both) perform(s) a vector product in
$(\mathbb{F}_q)^{2n+g-1}$ and subsequently a matrix-vector product, 
except for the first and the last step for which only
a matrix-vector product has to be performed. 
At step 0, the set $P_1$ (or $P_2$) computes $t_1=T(x)$. 
The set $P_1$ computes $t_{2^i}=T_1(t_{i-1}\cdot t_{i-1})$ at step $i=1 \ldots 3$. 
For its part, the set
$P_2$ computes $t_3=T_1(t_1 \cdot t_2)$ at step 2, $t_7=T_1(t_3 \cdot t_4)$ 
at step 3 and $t_{15}=T_1(t_8 \cdot t_7)$ 
at step 4. A last step is needed to retrieve the result $x^{15}=P \circ T^{-1}(t_{15})$.
Then we can  remark that three products in $(\mathbb{F}_q)^{2n+g-1}$ are made in parallel, to which must be added
the product performed by $P_1$ at the step 1, for an overall computation time of four 
products in $(\mathbb{F}_q)^{2n+g-1}$.

 \begin{figure}[t]
 \begin{center}
 
 \includegraphics[width=0.9 \linewidth]{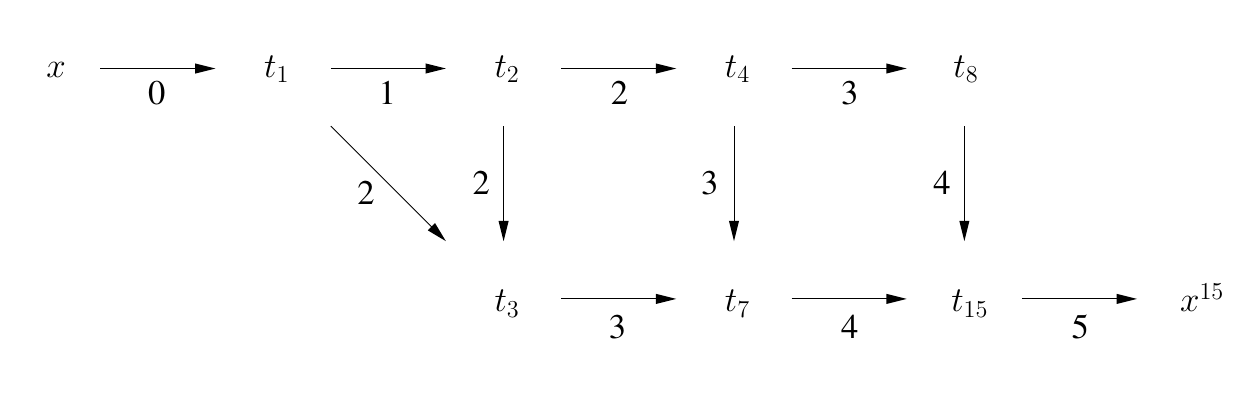}
 
 \caption{Diagram depicting the steps in the calculation of $x^{15}$. The initial and 
 latest steps are special since they involve only a matrix-vector product 
 (scalar multiplications).}
 \label{exponentiation_wo_precomput}
 \end{center}
 \end{figure}

We now consider operations over $\F_q$ and the amount of processing units used. 
In the model NS, the computation time is in $\BigO{n}$  with only $\BigO{n}$ processors. 
In the model S1, the computation time is in $\BigO{n}$ with $\BigO{n^2}$ processors
and in the model S2, the computation time is in $\BigO{n\log n}$ with $\BigO{n^2/\log n}$ processors.

\subsection{Exponentiation algorithm based on our model}\label{ChudNB}
With the use of a normal basis, the base $x$ does not have to be fixed anymore. Moreover,  the
 number of multiplications in $\F_{q^n}$ is reduced and it becomes possible to obtain an "overall" 
 parallel time in $\BigO{\log n}$ without prior storage. 

To make use of normal bases, let us now consider the $q$-ary representation $K$ of the exponent (composed of $n$ 
terms according to Fermat's Little Theorem) 
by writing
$$k=\sum_{i=0}^{n-1} K[i]q^i=\sum_{i=0}^{n-1} \sum_{j=0}^{t} K[i,j]2^jq^i,$$
where $K[i,j]$ for $j=0,\ldots,t$ are the bits of $K[i]$. Let $\sigma$ be the function such that  $\sigma(x,i)$ 
right shifts $i$ times the vector $x$. 
Then, we can express $x^k$ as:
\begin{align*}
\prod_{i=0}^{n-1}x^{K[i]q^i}& = & \prod_{i=0}^{n-1}\sigma(x,i)^{K[i]} & = & \prod_{i=0}^{n-1}\prod_{j=0}^{t}\sigma(x,i)^{K[i,j]2^j}\\
 & = & \prod_{i=0}^{n-1}\sigma(x^{K[i]},i) &= & \prod_{i=0}^{n-1}\sigma\left(\prod_{j=0}^{t}x^{K[i,j]2^j},i\right).
\end{align*}
This gives  two ways of rewriting the Square and Multiply algorithm. In fact, since we are no longer limited 
to the use of two (sets of) processors,
there exist more efficient algorithms, as shown by von zur Gathen \cite{Gat91,Gat92exp}. The idea is that short 
patterns might occur repeatedly in the $q-$ary 
representation of the exponent $k$.
Hence, precomputation of all patterns of a given short length $r\geq 1$ allows the overall cost to be lower.
By setting $s=\lceil n/r \rceil$ and writing $k=\sum_{0\leq i<s}k_iq^{ri}$ with $0\leq k_i<q^r$ for all $i$, von zur 
Gathen obtains optimal depth for an appropriate choice of $r$. In particular it is shown that 
this result is reachable 
with width in $\BigO{n/\log n}$ processors.

\begin{sloppypar}
Lee et al. \cite{LeeKPC05} introduced a rescheduling technique to reduce the number of processors at the counterpart 
of near-optimal depth. 
In what follows, we describe our own variant of von zur Gathen algorithm, 
achieving the same asymptotical efficiencies than the one from Lee et al., while being simpler.
We set $r=\lceil \log^2_q n - 2 \log_q(n) \log_q \log_q n \rceil$ and $u=\lfloor \log_q n - 2 \log_q \log_q n \rfloor$, and
rewrite the exponent in the following form:
$$k=\sum_{i=0}^{s-1} \left( \sum_{j=0}^{t-1} K_{i,j} q^{uj} \right)q^{ri}.$$
Thus, $s=\lceil n/r \rceil$ is in $\BigO{n/\log^2 n}$ and $t=\lceil r/u \rceil$ is in $\BigO{\log n}$. 
The algorithm is divided into five steps:
\begin{enumerate}
 \item for $2\leq l <q^u$, compute $x^l$,
 \item for $0\leq i <s$ and $0 \leq j < t$, compute $y_{i,j}=\sigma\left(x^{K_{i,j}},uj\right)$,
 \item for $0\leq i <s$, compute $y_{i}=\prod_{j=0}^{t-1}y_{i,j}$,
 \item for $0\leq i <s$, compute $z_{i}=\sigma\left(y_i,ri\right)$,
 \item return $x^k=\prod_{i=0}^{s-1}z_i$.
\end{enumerate}
We first examine the parallel time complexities in terms of multiplication in $\mathbb{F}_{q^n}$. 
We can use a binary tree of multiplications (executed from root to leaves) to
compute Step 1 in $\lceil \log_2(q^u-1) \rceil$ multiplications using 
$\max\{2^{\lceil \log_2(q^u-1) \rceil-2},q^u-1-2^{\lceil \log_2(q^u-1) \rceil-1}\}$ processors. 
Step 2 is free. In step 3, each $y_i$ for $0\leq i <s$ can be computed by distinct processors
in $t-1$ multiplications. Step 4 is free. 
Step 5 can be computed with a binary multiplication tree 
(executed from leaves to root) in $\lceil \log_2 s\rceil$ multiplications using $\lfloor s/2 \rfloor$ processors.
By summing the times of each step, the overall depth can be upper bounded by
\begin{multline*}
\left\lceil \log_2\frac{n}{\log_q^2n} \right\rceil + 
\left\lceil \frac{\log_qn+1}{\log_qn-2\log_q\log_qn-1} \right\rceil + \\
\left\lceil \log_qn \right\rceil + \left\lceil \log_2 \left(\frac{n}{\log_q(n/\log_q^2n)\log_qn}+1\right) \right\rceil
\end{multline*}
for sufficiently large $n$.
Moreover, by using an optimisation from von zur Gathen, the parallel execution time of Step 1 can be reduced to 
approximately $\log_2q+\log_2\log_q n$.  Thus, the overall depth is in $\BigO{\log n}$ multiplications. 
It can be noticed that at each parallel step, the number of processors involved 
stays in  $\BigO{n/\log^2 n}$. Consequently, this overall depth is achieved with a width in $\BigO{n/\log^2 n}$ processors.
\end{sloppypar}

\par From now on, we scale up the number of processors to optimize the running time in terms of operations in $\F_q$.
When considering this algorithm in our  Chudnovsky model, we consider $\BigO{n/\log^2 n}$ sets of $\BigO{n}$ processors 
if the scalar multiplications in $\mathbb{F}_q$ are considered free. Otherwise, we consider sets of $\BigO{n^2}$ processors. 

\begin{algorithm}[h]
\caption{Precomputation (modified von zur Gathen)
\newline Parallel works assigned to the sets of processors $(P_i)_{i=1..l}$ 
\newline where $l=\max(2^{\lceil \log_2(q^u-1) \rceil-2},q^u-1-2^{\lceil \log_2(q^u-1) \rceil-1})$}
\label{PrecomputeVzG}
\begin{algorithmic}
\Require $x,q$
\Ensure $x^d$ for $2\leq d<q^u $
\State $x_{1} \leftarrow x$
\State $h \leftarrow \lceil \log_2(q^u-1)\rceil$
\ForAll {$i \in \{ 1, \ldots, h-1\}$}
\ForAll {$j \in \{ 2^i, \ldots, 2^{i+1}-1 \}$\textbf{ the set} $P_{j-2^{i}+1}$}
\State $x_{j} \leftarrow P\circ T^{-1}(T(x_{\lceil j/2 \rceil} ).T(x_{\lfloor j/2 \rfloor}))$ 
  \EndFor
\EndFor
  \ForAll {$j \in \{ 2^{h}, \ldots,q^{u}-1\}$\textbf{ the set} $P_{j-2^{\lceil \log_2(q^u-1)\rceil}+1}$}
\State $x_{j} \leftarrow P\circ T^{-1}(T(x_{\lceil j/2 \rceil} ).T(x_{\lfloor j/2 \rfloor}))$ 
  \EndFor
  \State \textbf{return} $(x_i)_{i\in \{2,\ldots , q^u-1\}}$
\end{algorithmic}
\end{algorithm}

In terms of bilinear multiplications in $\F_{q}$, Step 1, described 
in Algorithm \ref{PrecomputeVzG}, is performed in depth $\BigO{\log n}$ and width $\BigO{n^2/\log^2 n}$. 
Step 2 and 4 consist in the shift of the first $n$ coordinates and are thus considered free.
Step 3, the details of which are left to the reader, is performed in depth $\BigO{\log n}$ and width $\BigO{n^2/\log^2 n}$.
The last step, which consists in computing $x^k=\prod_{0\leq i<s}z_i$, is a binary multiplication tree.
It is executed for a cost of $\BigO{\log n}$ bilinear multiplications using $\BigO{n^2/\log^2 n}$ processors. 
Overall,  in the model NS, this algorithm is done 
in depth $\BigO{\log n}$, width $\BigO{n^2/\log^2 n}$,  and size $\BigO{n^2/\log n}$, as stated in Theorem~\ref{LNS}. 
We let the reader deduce that, $i$) in the model S1, in terms of multiplications in $\mathbb{F}_q$, this algorithm 
is done in depth $\BigO{\log n}$, width $\BigO{n^3/\log^2 n}$,  and size $\BigO{n^3/\log n}$;
$ii$) in the model S2, in terms of any operations in $\mathbb{F}_q$, this algorithm is done in
depth $\BigO{\log^2 n}$, width $\BigO{n^3/\log^3 n}$,  and size $\BigO{n^3/\log n}$.~\\

\paragraph{\textbf{Reducing the number of scalar operations. }} 
We have seen that in the scalar model, the high number of processors 
is due to the high number of scalar operations (in particular the matrix-vector products $Tx$, $T_1x$ and $T^{-1}x$). 
At each step of the computation, 
the number of multiplications in $\mathbb{F}_{q^n}$ done in parallel is in $\BigO{n/\log^2 n}$ in the worst case. 
Thus, at a step of the computation involving
multiple parallel matrix-vector products, 
instead of performing 
$\BigO{n/\log^2 n}$ separate matrix-vector products, we  write the $\BigO{n/\log^2 n}$ vectors as 
columns of a matrix $B$, then complete this matrix with zero
columns in order to obtain a square matrix. Now, we have a product of two square matrices that 
can be performed using  the Coppersmith-Winograd method, thus reducing the number of scalar operations. 
In the S1 model, the Coppersmith-Winograd product can be performed in  $\BigO{1}$ multiplications in $\F_q$
using $\BigO{n^{2.38}}$ processors. This optimization allows the exponentiation to be done in 
depth $\BigO{\log n}$, width $\BigO{n^{2.38}}$ and size $\BigO{n^{2.38}\log n}$.  
In the model S2, the Coppersmith-Winograd product can be performed in  $\BigO{\log n}$ operations in $\F_q$
using $\BigO{n^{2.38}/\log n}$ processors. The optimization allows the exponentiation to be done in 
depth $\BigO{\log^2 n}$, width $\BigO{n^{2.38}/\log n}$ and size $\BigO{n^{2.38}\log n}$.  

Remark that the number of zero columns in the matrix $B$ may not be negligible.
Thus, instead of filling  $B$ with zero columns in order to obtain a square matrix, we can slightly improve the complexity by using 
the CW method in the following way: Consider a matrix $B$ containing only
$\BigTheta{n/\log^2 n}$ vectors. Let us denote by $L$ the number of columns of $B$ 
($L$ is then in $\BigTheta{n/\log^2 n}$). We partition $B$ in square submatrices of size $L$ 
so that $B$ is seen as a column $(B_1, B_2, \ldots, B_j)^T$ of blocks (with $j$ in $\BigO{\log^2 n}$). 
In the same way, we partition $T_1$ in square submatrices of size $L$ so that a row $i$ of blocks of $T_1$ 
is represented as $(A_{i1}, A_{i2}, \ldots, A_{ij})$.
The sum of products $\sum_{k=1}^j A_{ik}B_k$ corresponds to the $i$-th block of the resulting column of blocks.
A product $A_{ik}B_k$ is done using CW in $\BigO{(\frac{n}{\log^2n})^{2.38}}$ scalar multiplications. 
Since we have $\BigO{\log^4 n}$ such products to compute $T_1B$, the 
overall number of scalar multiplications is in $\BigO{n^{2.38}/\log^{0.76}n}$. 
Consequently, in the model S1 the product $T_1B$ can be computed in $\BigO{1}$ multiplications in $\mathbb{F}_q$ 
using $\BigO{n^{2.38}/\log^{0.76}n}$ processors, whereas in the model S2, it can be computed in $\BigO{\log n}$ 
operations in $\mathbb{F}_q$ 
using $\BigO{n^{2.38}/\log^{1.76}n}$ processors.

We can substitute these last results in the case of the parallel exponentiation to obtain the stated complexities
of Theorem~\ref{LS1} and Theorem~\ref{LS2}, with the  current exponent for the best optimized matrix product:  in the model S1,
$x^k$ can be computed in $\BigO{\log n}$  multiplications in $\mathbb{F}_q$ using $\BigO{n^{2.38}/\log^{0.76}n}$ 
processors for a size of $\BigO{n^{2.38}\log^{0.24}n}$  multiplications in $\mathbb{F}_q$, 
whereas in the model S2, 
$x^k$ can be computed in  $\BigO{\log^2 n}$ 
operations in $\mathbb{F}_q$ using $\BigO{n^{2.38}/\log^{1.76}n}$ processors for a size of $\BigO{n^{2.38}\log^{0.24}n}$ 
operations in $\mathbb{F}_q$.

 \section{Multiplication in $\F_{16^{n}/\F_{16}}$}\label{effective}
 
 Set $q=16$ and $n=13, 14, 15$. From now on, $F/\F_q$ denotes the algebraic function field
 associated to the hyper elliptic curve $X$ with plane model $y^2+y=x^5$, of genus two. This curve has 33 
 rational points, which is maximal over $\F_q$ according to the Hasse-Weil bound. We represent $\F_{16}$  
 as the field $\F_2(a)=\F_2[X]/(P(X))$ where $P(X)$ is the irreducible polynomial $P(X)= X^4+X+1$ and 
 $a$ denotes a primitive root of $P(X)= X^4+X+1$. Let us give the projective coordinates $(x:y:z)$ of 
 rational points of the curve $X$:
 
 \vspace{.5em}
 
 $$
 \begin{array}{lll}
 P_{\infty}=(0: 1 : 0) & P_{2}=(0: 0: 1) & P_{3}=( 0: 1: 1) \\
 P_{4}=(a : a: 1)& P_{5}=( a: a^4: 1) & P_{6}=(a^2: a^2: 1) \\
 P_{7}=( a^2: a^8: 1) & P_{8}=( a^3: a^5: 1)& P_{9}=( a^3: a^{10}: 1) \\
 P_{10}=(a^4: a: 1) & P_{11}=( a^4: a^4: 1) & P_{12}=( a^5: a^2: 1) \\
 P_{13}=(a^5 : a^8: 1) & P_{14}=(a^6: a^5: 1) & P_{15}=( a^6: a^{10}: 1) \\
 P_{16}=( a^7: a: 1) & P_{17}=( a^7: a^4: 1)& P_{18}=(a^8: a^2: 1) \\
 P_{19}=( a^8: a^8: 1) & P_{20}=( a^9: a^5: 1) & P_{21}=( a^9:a^{10} : 1) \\
 P_{22}=(a^{10}: a: 1) & P_{23}=( a^{10}: a^4: 1) & P_{24}=( a^{11}: a^2: 1) \\
 P_{25}=( a^{11}:a^8 :1 ) & P_{26}=(a^{12}:a^5 : 1) &P_{27}=( a^{12}: a^{10}: 1) \\
 P_{28}=( a^{13}: a: 1) & P_{29}=( a^{13}: a^4: 1) & P_{30}=(a^{14}: a^2: 1) \\
 P_{31}=( a^{14}: a^8: 1) & P_{32}=( 1: a^5: 1) & P_{33}=( 1: a^{10}: 1) \\
 \end{array}
 $$
 
 \subsection{Construction of the required divisors}
 
 \subsubsection{A place $Q$ of degree n}\label{Q13}
 
 It is sufficient to take a place $\mathcal{Q}$ of degree $n$ in the rational function field $\F_{q}(x)/\F_q$, 
 which totally splits in $F/\F_q$. It is equivalent to choose a monic irreducible polynomial 
 $\mathcal{Q}(x)\in \F_{q}[x]$ of degree n such that its roots $\alpha_i$ in $\F_{q^n}$ satisfy $Tr_{\F_{2}}(\alpha_i^5)=0$ 
 for $i=1,...,n$ where the map $Tr_{\F_{2}}$ denotes the classical function Trace over $\F_{2}$ by \cite[Theorem 2.25]{lini2}.
 In fact, it is sufficient to verify that this property is satisfied for only one root since a finite field is Galois.
 Moreover, in the context of our method, we require that this irreducible polynomial $\mathcal{Q}(x)$ corresponds
 to a normal polynomial (cf. Section \ref{baseFQ}). 

 \vspace{1em}
 
 For example, for the extension $n=13$, we choose the primitive normal polynomial 
 
 \begin{equation}
 \begin{aligned}
 \mathcal{Q}(x)= & x^{13}+a^{6}x^{12}+a^{5}x^{11}+a^{11}x^{10}+x^{9}+a^{12}x^{8}+\\
 & a^{7}x^{7}+a^{7}x^{5}+a^{2}x^{4}+a^{11}x^{3}+a^{8}x^{2}+a^{6}x+a^{14} .
 \end{aligned}
 \end{equation}
 
 Let $b$ be one primitive root of  $\mathcal{Q}(x)$. It is easy to check that $Tr_{\F_{2}}(b^5)=0$, hence 
 the place $(\mathcal{Q}(x))$ of $\F_{16}(x)/\F_{16}$ 
 is totally splitted in the algebraic function field $F/\F_q$, which means that there exist 
 two places of degree $n$ in $F/\F_q$ lying over the place $(\mathcal{Q}(x))$ of 
 $\F_{16}(x)/\F_{16}$. 
  For the place $Q$ of degree $n$ in the algebraic function field $F/\F_q$, we consider one of the two 
  places in $F/\F_q$ lying over the place 
 $(\mathcal{Q}(x))$ of $\F_{16}(x)/\F_{16}$, namely the orbit of the $\F_{16^{13}}$-rational point 
 $\mathcal{P}_{1i}=(\alpha_i, \beta_i:1)$ where $\alpha_i$ 
 is a root of $\mathcal{Q}(x)$ and 
 $\beta_i= a^6\alpha_i^{12} + a^{13}\alpha_i^{11} + a\alpha_i^{10} + a^{13}\alpha_i^9 + a^8\alpha_i^8 + a\alpha_i^7 + a^8\alpha_i^6 +
   a^9\alpha_i^5 + a^5\alpha_i^4 + a^2\alpha_i^2 + a^8\alpha_i + a^{13}$ for $i=1,...,13$. 
   Notice that the second place is given by the conjugated points 
   $\mathcal{P}_{2i}=(\alpha_i: \beta_i+1:1)$ for $i=1,...,13$.
 
  \vspace{1em}
  
   \subsubsection{A place $D$ of degree n+g-1}
 
 For the divisor $D$ of degree $n+g-1$, we choose a place $D$ of degree 14 according to the method used for the place $Q$. 
 We consider the orbit of the $\F_{16^{14}}$-rational point $\mathcal{P}_{1i}=(\gamma_i, \delta_i:1)$ where $\gamma_i$ 
 is a root of $\mathcal{D}(x)=x^{14} + a^9x^{13} + a^6x^{12} + a^7x^{11} + a^{11}x^{10} + a^{12}x^9 + a^{10}x^8 +
   a^6x^7 + a^7x^6 + a^{10}x^5 + a^{14}x^4 + x^3 + x^2 + a^3x + a$
 and $\delta_i=a^4\gamma_i^{12} + a^8\gamma_i^{11} + a^7\gamma_i^9 + a^2\gamma_i^8 + a^3\gamma_i^7 + a^8\gamma_i^6 
+ a^4\gamma_i^5 +a^{14}\gamma_i^4 + \gamma_i^2 + a^6\gamma_i + a^3$ for $i=1,...,14$.
 Notice that the second place is given by the conjugated points $\mathcal{T}_{2i}=(\gamma_i, \delta_i+1:1)$ for $i=1,...,14$.
 
  \vspace{1em}
  
  The place $Q$ and the divisor $D$ satisfy the good properties since the dimension of the divisor $D-Q$ is 
  zero which means that the divisor  
  $D-Q$ is non-special of degree $g-1$.
  
  \subsubsection{The basis of the residue class field $F_Q$}
  
  We choose as basis of the residue class field $F_Q$ the normal basis ${\mathcal B}_Q$ associated to the place $Q$ 
  obtained in Section \ref{Q13}.

  \subsubsection{The basis of $\mathcal{L}(D)$}
  
  We choose as basis of the Riemann-Roch space $\mathcal{L}(D)$ the basis ${\mathcal B}_D=(f_1,...,f_n)$ 
  such that $E({\mathcal B}_D)={\mathcal B}_Q$ 
  is a normal basis of $F_Q$ as in Section \ref{baseLD}. 
  Any element $f_i$ of ${\mathcal B}_D$ is such that $$f_i(x,y)=\frac{f_{i1}(x)y+f_{i2}(x)}{\mathcal{D}(x)},$$
  where $f_{i1}, f_{i2}\in \F_{16}[x]$. To simplify, we set $f_i(x,y)=(f_{i1}(x), f_{i2}(x))$. 
  Let us give the elements of ${\mathcal B}_D$:
  
  \vspace{1em}
  
$f_1(x,y)=(a^{13}x^{11} + a^{10}x^{10} + a^3x^9 + a^{10}x^8 + a^{14}x^7 + a^{11}x^6 + a^8x^5 + a^{11}x^4 +
   x^3 + ax^2 + a^{11}x + a^{11}, a^{12}x^{14} + a^{12}x^{13} + a^9x^{12} + x^{11} + a^8x^{10} + a^{13}x^9 + a^{12}x^8 +
   ax^7 + a^5x^6 + x^5 + a^{13}x^4 + a^5x^3 + a^{12}x^2 + a^4x)$,\\
      
$f_2(x,y)=(a^{11}x^{11} + a^5x^{10} + x^9 + ax^8 + a^{14}x^7 + a^{11}x^6 + a^2x^5 + a^4x^4 + a^7x^3
   + a^7, a^8x^{14} + a^7x^{13} + a^{12}x^{12} + ax^{11} + a^3x^{10} + a^7x^9 + a^{10}x^8 + a^9x^7 + a^{12}x^6 +
   a^{11}x^5 + a^{13}x^4 + a^{14}x^3 + a^{13}x^2 + a^2x + a^{12})$,\\
   
$f_3(x,y)=(a^4x^{11} + a^8x^{10} + a^8x^9 + x^8 + a^2x^7 + a^{14}x^6 + a^2x^5 + a^4x^4 + a^{12}x^3
   + a^3x^2 + a^4, ax^{14} + a^7x^{13} + a^6x^{12} + ax^{11} + a^9x^{10} + a^{11}x^9 + a^7x^8 + a^3x^7 +
   a^7x^6 + ax^5 + a^7x^4 + a^{13}x^3 + a^2x^2 + a^8x)$,\\
        
$f_4(x,y)=(a^5x^{11} + a^{13}x^{10} + a^2x^9 + a^8x^8 + a^9x^7 + a^6x^6 + a^2x^4 + a^6x^3 +
   a^{13}x^2 + a^9x + a^{11}, a^3x^{14} + a^{12}x^{13} + a^5x^{12} + a^6x^{11} + a^{11}x^{10} + a^3x^9 + a^5x^8 +
   a^2x^7 + a^{11}x^6 + a^2x^5 + a^{11}x^4 + a^7x^3 + ax^2 + a^4x + a^{12})$,\\  
   
$f_5(x,y)=(a^4x^{11} + a^2x^9 + ax^8 + a^{13}x^7 + a^{12}x^6 + x^5 + ax^4 + a^{13}x^3 + a^{14}x^2 +
   ax + a^6, a^{11}x^{14} + a^{10}x^{13} + x^{12} + a^{12}x^{11} + a^3x^{10} + a^{12}x^9 + a^9x^8 + a^4x^7 + a^{14}x^6 +
   a^2x^5 + a^{11}x^4 + a^{11}x^3 + a^6x^2 + a^5x + a^3)$,\\
   
$f_6(x,y)=(a^6x^{11} + a^{12}x^{10} + a^{10}x^9 + a^7x^8 + a^8x^7 + a^6x^5 + x^4 + a^{13}x^3 +
   a^8x^2 + 1, a^{10}x^{14} + a^4x^{13} + x^{12} + a^4x^{11} + a^2x^{10} + a^7x^9 + a^5x^8 + a^{13}x^7 + ax^6 +
   a^6x^5 + a^9x^4 + a^7x^3 + a^8x^2 + a^2x + a^{11})$,\\   
      
$f_7(x,y)=(x^{11} + a^5x^{10} + ax^9 + ax^8 + a^{10}x^7 + a^{12}x^6 + a^{14}x^5 + a^3x^4 + a^3x^3 +
   x^2 + a^4x + a^9, a^7x^{14} + a^3x^{13} + a^8x^{12} + a^4x^{11} + a^6x^{10} + a^6x^9 + a^5x^8 + a^3x^7 +
   a^{13}x^6 + a^6x^5 + a^4x^4 + a^{14}x^3 + a^{11}x^2 + a^{11})$,\\  
     
$f_8(x,y)=(a^6x^{11} + a^6x^{10} + a^{12}x^9 + x^8 + a^4x^6 + ax^5 + a^{11}x^4 + x^3 + ax^2 +
   a^{13}x + a^3, a^9x^{14} + a^7x^{13} + a^{14}x^{12} + a^9x^{11} + a^7x^{10} + a^8x^9 + a^{13}x^8 +
   a^{10}x^7 + a^{10}x^5 + a^{10}x^4 + a^9x^3 + a^7x^2 + ax + a)$, \\ 
     
$f_9(x,y)=(x^{11} + a^{12}x^{10} + a^{13}x^9 + a^{14}x^8 + a^{13}x^6 + a^{14}x^5 + a^4x^4 + a^{11}x^3 +
   a^2x^2 + x + a^{11}, x^{14} + a^{14}x^{13} + a^5x^{12} + a^5x^{11} + ax^{10} + a^6x^9 + a^3x^8 + a^{11}x^7 + a^{11}x^6
   + a^5x^5 + a^{12}x^4 + x^3 + a^{13}x^2 + a^6x + a^4)$,\\   
      
$f_{10}(x,y)=(a^5x^{11} + a^5x^{10} + a^4x^9 + ax^8 + a^9x^7 + a^5x^6 + a^5x^5 + a^4x^4 + a^6x^3
   + a^{14}x^2 + a^7x + a^{12}, a^2x^{14} + a^7x^{13} + a^{10}x^{11} + a^9x^{10} + a^{14}x^8 + x^7 + x^6 + a^5x^5 +
   a^{11}x^4 + a^6x^3 + a^{12}x^2 + a^6x + a^{13})$,\\
      
$f_{11}(x,y)=(a^{14}x^{11} + a^{14}x^{10} + a^{13}x^9 + a^{11}x^8 + a^7x^7 + a^9x^6 + a^{11}x^5 + a^3x^4 +
   a^6x^3 + a^8x^2 + x + a^9, a^{12}x^{14} + a^4x^{13} + a^5x^{12} + a^5x^{11} + a^{10}x^{10} + a^{10}x^9 + a^{12}x^8
   + a^6x^7 + ax^6 + a^2x^4 + a^9x^2 + a^{11}x + a^{10})$,\\
     
$f_{12}(x,y)=(a^8x^{11} + a^9x^{10} + a^2x^9 + a^3x^8 + ax^7 + a^{14}x^5 + x^4 + a^{11}x^3 + a^3x^2 +
   ax + a^{12}, a^{13}x^{14} + a^{13}x^{13} + a^8x^{12} + a^{14}x^{11} + a^4x^{10} + a^{11}x^9 + a^3x^8 + a^{10}x^6 + a^6x^5
   + a^9x^3 + a^{14}x^2 + a^{10}x + a^3)$, \\  
      
$f_{13}(x,y)=(a^{12}x^{11} + a^5x^{10} + x^9 + a^5x^8 + ax^7 + a^{10}x^6 + a^7x^5 + a^5x^4 + a^3x^2 +
   x + a^4, a^4x^{14} + a^9x^{13} + a^{14}x^{12} + a^7x^{11} + a^9x^{10} + a^5x^8 + a^{14}x^7 + a^{13}x^6 + a^2x^5
   + a^5x^4 + a^2x^3 + a^8x^2 + ax)$,\\    
       
 \vspace{1em}
 
 \subsubsection{The basis of $\mathcal{L}(2D)$}
 
 Let ${\mathcal B}_{2D}=(g_1,...,g_{2n+g-1})$ be a basis of $\mathcal{L}(2D)$ as defined in Section \ref{baseL2D}. 
 Since the divisor $D$ is effective, we can complete the basis $f$ of  $\mathcal{L}(D)$ in $\mathcal{L}(2D)$. 
 Then any element $g_i$ of $g$ 
 is such that $g_i(x,y)=f_i(x,y)$ for $i=1,...,n$ and $g_i(x,y)=\frac{g_{i1}(x)y+g_{i2}(x)}{\mathcal{D}^2(x)}$
 (which we will denote $g_i(x,y)=(g_{i1}, g_{i2})$), with $g_{i1}, g_{i2} \in \F_{16}[x]$ for $i=n+1,...,2n+g-1$.
 Moreover, we require the completion of $\mathcal{L}(D)$ by the kernel of the map $P$. 
 Hence, the set $(g_i)_{i=n+1,...,2n+g-1}$ is a basis of the kernel of the restriction map over $\mathcal{L}(2D)$ 
 of the canonical projection from the valuation ring of the place  $Q$ in its residue class field $F_Q$.
 
 \vspace{1em}
 
 $g_{14}(x,y)=(a^{13}x^{24} + a^7x^{23} + a^2x^{22} + a^7x^{21} + a^{14}x^{20} + a^2x^{19} + a^9x^{18} +
   a^6x^{17} + a^{10}x^{16} + a^6x^{15}+ a^{12}x^{14}+ x^{13} + a^{12}x^{12} + a^3x^{11} +
   a^8x^9 + a^4x^8 + a^8x^7 + a^9x^6 + a^{12}x^5 + a^{14}x^4 + x^3 + a^8x^2
   + x + a^9, a^7x^{26} + ax^{25} + a^{10}x^{23} + a^{14}x^{22} + a^4x^{21} + a^8x^{20}
   + a^7x^{18} + a^9x^{17} + ax^{16} + a^{12},x^{15}+ a^{11}x^{14}+ a^2x^{13} + a^9x^{12}
   + a^{10}x^{11} + a^9x^{10} + ax^8 + a^{11}x^7 + a^5x^6 + a^6x^5 + a^5x^4 +
   a^{14}x^3 + a^5x^2 + a^{11}x + a^8)$,\\

$g_{15}(x,y)=(a^{13}x^{25} + a^7x^{24} + a^2x^{23} + a^7x^{22} + a^{14}x^{21} + a^2x^{20} + a^9x^{19} +
   a^6x^{18} + a^{10}x^{17} + a^6x^{16} + a^{12}x^{15}+ x^{14}+ a^{12}x^{13} + a^3x^{12} +
   a^8x^{10} + a^4x^9 + a^8x^8 + a^9x^7 + a^{12}x^6 + a^{14}x^5 + x^4 + a^8x^3
   + x^2 + a^9x,a^7x^{27} + ax^{26} + a^{10}x^{24} + a^{14}x^{23} +
   a^4x^{22} + a^8x^{21} + a^7x^{19} + a^9x^{18} + ax^{17} + a^{12}x^{16} + a^{11}x^{15}+
   a^2x^{14}+ a^9x^{13} + a^{10}x^{12} + a^9x^{11} + ax^9 + a^{11}x^8 + a^5x^7 +
   a^6x^6 + a^5x^5 + a^{14}x^4 + a^5x^3 + a^{11}x^2 + a^8x)$,\\
   
$g_{16}(x,y)=(a^7x^{25} + a^{10}x^{24} + a^9x^{23} + a^{11}x^{22} + a^4x^{21} + x^{20} + a^5x^{19} +
   ax^{18} + a^{14}x^{17} + a^6x^{16} + a^4x^{15}+ a^5x^{14}+ a^9x^{13} + a^{10}x^{12} +
   a^3x^{11} + a^{14}x^{10} + a^{10}x^9 + a^4x^8 + a^8x^7 + a^{10}x^6 + x^5 +
   a^5x^4 + a^4x^2 + a^4x + a, ax^{27} + a^4x^{26} + a^{11}x^{25} + a^{13}x^{24}
   + x^{23} + a^{12}x^{22} + a^8x^{20} + a^{11}x^{19} + a^4x^{18} + ax^{17} + a^8x^{16} +
   a^8x^{15}+ a^2x^{14}+ a^9x^{13} + a^{14}x^{12} + a^2x^{11} + a^5x^{10} + a^3x^9 +
   a^{11}x^8 + a^2x^7 + a^{12}x^6 + a^{10}x^5 + ax^4 + a^9x^3 + a^5x^2 + x +
   a^7)$,\\
   
$g_{17}(x,y)=(a^{10}x^{25} + a^{10}x^{24} + a^2x^{23} + ax^{21} + a^2x^{20} + a^9x^{19} + a^{13}x^{18} +
   a^3x^{17} + a^{11}x^{16} + x^{15}+ a^{12}x^{14}+ a^7x^{12} + ax^{11} + a^{12}x^{10} +
   a^2x^9 + a^{12}x^7 + a^{10}x^6 + a^5x^5 + a^{13}x^4 + a^{14}x^3 + a^5x^2 +
   a^{12}x + a^{10},a^4x^{27} + a^4x^{26} + a^3x^{25} + a^{12}x^{24} +
   a^3x^{23} + a^3x^{22} + a^8x^{21} + a^3x^{20} + a^{13}x^{19} + a^3x^{18} + a^{11}x^{17}
   + a^{10}x^{16} + a^{11}x^{15}+ a^6x^{14}+ ax^{13} + a^{14}x^{12} + a^3x^{11} +
   a^5x^{10} + a^{10}x^9 + a^7x^8 + ax^6 + a^2x^5 + a^2x^4 + a^3x^3 +
   a^6x^2 + x + a)$,\\
   
$g_{18}(x,y)=(a^{10}x^{25} + a^5x^{24} + a^{12}x^{23} + a^{14}x^{22} + a^{12}x^{21} + a^{14}x^{20} + x^{19} +
   a^{11}x^{18} + a^3x^{17} + a^4x^{16} + ax^{15}+ a^{11}x^{14}+ a^5x^{13} + a^{14}x^{12}
   + a^3x^{11} + a^3x^{10} + a^{13}x^9 + x^8 + a^7x^7 + a^4x^6 + a^{13}x^5 +
   a^7x^4 + a^{14}x^3 + x^2 + a^8x + a^{13}, a^4x^{27} + a^9x^{26} +
   x^{25} + x^{24} + a^8x^{23} + a^{14}x^{22} + a^3x^{21} + a^3x^{20} + x^{19} + x^{18} +
   a^{13}x^{16} + a^4x^{15}+ a^{10}x^{14}+ a^{11}x^{13} + a^9x^{12} + a^{12}x^{11} +
   a^{11}x^{10} + a^{12}x^9 + x^8 + a^4x^7 + a^9x^6 + x^5 + x^4 + a^8x^3 +
   a^4x^2 + a^{13}x + a^4)$,\\
   
$g_{19}(x,y)=(a^5x^{25} + a^{13}x^{24} + a^5x^{23} + a^2x^{22} + ax^{21} + ax^{20} + ax^{19} +
   a^{11}x^{18} + a^8x^{17} + a^4x^{15}+ a^3x^{14}+ a^2x^{13} + a^4x^{12} + a^{12}x^{11}
   + x^{10} + a^6x^9 + a^8x^8 + a^{12}x^7 + a^3x^6 + a^7x^5 + a^7x^4 +
   a^{11}x^3 + a^7x^2 + a^{12}x + a^{13}, a^9x^{27} + a^4x^{26} + a^{12}x^{25} +
   a^6x^{24} + a^2x^{23} + a^2x^{22} + a^3x^{21} + a^2x^{20} + a^3x^{19} + a^{12}x^{18}
   + a^9x^{17} + a^3x^{15}+ a^7x^{14}+ a^{13}x^{13} + a^{13}x^{12} + a^{10}x^{11} +
   a^5x^{10} + a^8x^9 + ax^8 + a^7x^7 + a^{12}x^6 + a^2x^5 + a^6x^4 +
   a^9x^3 + a^{12}x^2 + a^6x + a^4)$,\\
   
$g_{20}(x,y)=(a^{13}x^{25} + a^3x^{24} + a^{12}x^{23} + a^5x^{22} + a^5x^{21} + a^7x^{20} + a^6x^{19} +
   a^2x^{18} + x^{17} + a^{13}x^{16} + a^{13}x^{15}+ a^3x^{14}+ a^5x^{13} + a^7x^{12} +
   a^{10}x^{11} + a^{11}x^{10} + a^4x^8 + a^9x^7 + a^4x^6 + a^7x^5 + a^4x^2 +
   a^4x + a^8, a^4x^{28} + a^4x^{27} + x^{26} + a^{11}x^{24} + x^{23} + a^9x^{22}
   + a^2x^{21} + a^8x^{19} + x^{18} + a^5x^{17} + x^{16} + a^8x^{14}+ a^7x^{13} +
   a^{14}x^{12} + a^6x^{11} + a^{12}x^{10} + a^{13}x^9 + a^6x^8 + a^3x^7 + a^3x^6 +
   a^2x^5 + a^8x^3 + ax^2 + a^3x + a^{14})$,\\
   
$g_{21}(x,y)=(a^8x^{25} + a^{13}x^{24} + a^{13}x^{23} + a^2x^{22} + a^{11}x^{21} + ax^{20} + a^5x^{19} +
   a^3x^{17} + a^{13}x^{16} + x^{15}+ ax^{14}+ a^2x^{13} + a^3x^{12} + ax^{11} +
   a^{10}x^{10} + x^9 + a^3x^8 + x^7 + a^{13}x^6 + a^2x^5 + a^{12}x^4 + x^3 +
   a^9x^2 + ax + a^4, a^2x^{27} + a^7x^{26} + a^{10}x^{25} + a^{14}x^{24} +
   a^9x^{23} + a^{14}x^{22} + a^6x^{21} + ax^{20} + a^5x^{19} + x^{18} + a^4x^{17} +
   a^{14}x^{16} + a^9x^{15}+ a^8x^{14}+ a^8x^{13} + a^{11}x^{12} + a^5x^{11} + a^7x^{10}
   + a^{12}x^8 + a^2x^7 + a^9x^6 + a^2x^5 + a^8x^4 + a^{14}x^3 + a^2x^2 +
   a^{13}x + a^{13})$,\\
   
$g_{22}(x,y)=(a^{13}x^{25} + a^2x^{24} + a^4x^{23} + a^3x^{22} + a^2x^{21} + ax^{20} + a^4x^{19} +
   a^8x^{17} + a^3x^{16} + a^6x^{15}+ a^{11}x^{14}+ a^5x^{13} + a^2x^{12} + ax^{11} +
   ax^{10} + a^5x^9 + a^7x^8 + a^{11}x^7 + a^3x^6 + a^{12}x^5 + a^3x^4 +
   a^{13}x^3 + a^3x^2 + a^9x + a^{11}, a^7x^{27} + a^{11}x^{26} + a^4x^{25} +
   a^8x^{24} + a^{10}x^{23} + a^{12}x^{22} + ax^{21} + a^9x^{20} + a^{11}x^{19} + a^2x^{18}
   + a^6x^{17} + a^{11}x^{16} + ax^{15}+ a^{14}x^{14}+ a^7x^{13} + a^{12}x^{12} +
   a^2x^{11} + a^{12}x^{10} + a^{11}x^9 + a^{14}x^8 + a^{10}x^7 + a^{11}x^6 + a^{14}x^5
   + a^5x^4 + a^7x^3 + a^2x^2 + a^9x + a^2)$,\\
   
$g_{23}(x,y)=(a^2x^{25} + a^{14}x^{23} + a^4x^{22} + a^8x^{21} + a^3x^{20} + a^9x^{19} + a^{13}x^{17} +
   a^{12}x^{16} + a^6x^{15}+ a^{12}x^{14}+ a^5x^{13} + a^8x^{12} + a^{12}x^{11} +
   a^6x^{10} + a^{14}x^9 + a^{10}x^8 + ax^7 + x^6 + a^3x^5 + a^{11}x^4 + a^{14}x^3
   + a^4x^2 + a^{13}x + a, a^{11}x^{27} + a^6x^{25} + a^4x^{24} + a^{13}x^{23} +
   a^3x^{22} + a^9x^{21} + a^{13}x^{18} + a^4x^{17} + a^{14}x^{16} + a^3x^{15}+ a^8x^{14}
   + ax^{13} + a^{10}x^{12} + a^7x^{11} + a^9x^{10} + a^{12}x^9 + a^{12}x^8 + a^5x^7 +
   a^3x^5 + a^{12}x^4 + a^{14}x^2 + a^8x + a^7)$,\\
   
$g_{24}(x,y)=(a^6x^{24} + a^6x^{22} + x^{21} + a^2x^{20} + a^{13}x^{19} + ax^{18} + a^{14}x^{16} +
   a^{14}x^{15}+ a^{11}x^{14}+ a^4x^{13} + a^5x^{12} + a^3x^{11} + a^2x^{10} + x^9 +
   a^9x^8 + a^6x^7 + a^{14}x^6 + a^{11}x^5 + a^6x^4 + a^8x^2 + a^{10}x +
   a^5, a^{14}x^{26} + a^7x^{25} + x^{24} + a^{11}x^{23} + a^{10}x^{22} + x^{20} +
   x^{19} + a^{13}x^{17} + a^5x^{16} + a^5x^{15}+ a^4x^{14}+ a^4x^{13} + a^{11}x^{12} +
   a^5x^{11} + a^4x^{10} + a^8x^9 + a^{13}x^8 + a^7x^7 + a^{11}x^5 + a^6x^4 +
   a^8x^3 + a^3x + a^{11})$,\\
   
$g_{25}(x,y)=[(a^6x^{25} + a^6x^{23} + x^{22} + a^2x^{21} + a^{13}x^{20} + ax^{19} + a^{14}x^{17} +
   a^{14}x^{16} + a^{11}x^{15}+ a^4x^{14}+ a^5x^{13} + a^3x^{12} + a^2x^{11} + x^{10} +
   a^9x^9 + a^6x^8 + a^{14}x^7 + a^{11}x^6 + a^6x^5 + a^8x^3 + a^{10}x^2 +
   a^5x, a^{14}x^{27} + a^7x^{26} + x^{25} + a^{11}x^{24} + a^{10}x^{23} + x^{21} +
   x^{20} + a^{13}x^{18} + a^5x^{17} + a^5x^{16} + a^4x^{15}+ a^4x^{14}+ a^{11}x^{13} +
   a^5x^{12} + a^4x^{11} + a^8x^{10} + a^{13}x^9 + a^7x^8 + a^{11}x^6 + a^6x^5 +
   a^8x^4 + a^3x^2 + a^{11}x)$,\\
   
$g_{26}(x,y)=(a^4x^{24} + a^2x^{23} + a^6x^{22} + a^8x^{21} + a^4x^{20} + a^2x^{19} + a^7x^{18} +
   a^7x^{17} + x^{16} + a^{14}x^{15}+ a^{13}x^{14}+ ax^{13} + a^6x^{12} + a^{13}x^{11} +
   a^{11}x^{10} + a^5x^9 + ax^8 + a^9x^7 + a^{12}x^6 + a^9x^4 + ax^3 +
   a^{13}x^2 + a^{12}x + a^9, a^3x^{28} + a^7x^{27} + a^{11}x^{26} + a^8x^{25} +
   a^9x^{23} + a^8x^{22} + x^{21} + a^4x^{20} + a^9x^{19} + a^4x^{18} + a^9x^{17} +
   ax^{16} + a^5x^{15}+ a^{13}x^{14}+ a^9x^{13} + a^6x^{12} + ax^{11} + a^9x^{10} +
   a^5x^9 + a^{11}x^8 + a^{10}x^6 + a^5x^5 + a^{10}x^4 + a^{12}x^3 + x^2 + a^8x
   + 1)$,\\
   
$g_{27}(x,y)=(a^2x^{25} + x^{24} + x^{23} + a^{14}x^{22} + a^6x^{21} + ax^{20} + a^{12}x^{19} + a^{14}x^{18}
   + a^{10}x^{17} + a^{13}x^{16} + a^5x^{14}+ a^5x^{13} + a^{14}x^{12} + a^9x^{11} +
   a^6x^{10} + a^{13}x^9 + a^{10}x^8 + a^8x^7 + a^3x^6 + a^{12}x^5 + a^7x^4 +
   x^3 + a^{10}x + a^8, a^6x^{28} + x^{27} + x^{26} + a^4x^{25} + a^6x^{24} +
   a^4x^{23} + a^{11}x^{21} + a^5x^{20} + a^{11}x^{19} + a^9x^{18} + a^8x^{17} + ax^{16} +
   a^{13}x^{15}+ a^4x^{14}+ x^{13} + a^3x^{11} + a^2x^{10} + a^4x^9 + a^5x^8 + x^7
   + a^3x^6 + a^{12}x^5 + a^{11}x^4 + a^6x^3 + a^{14}x^2 + a^{13}x + a^{14})$.\\

\appendix

\section{Magma implementation of the multiplication algorithm in the finite field 
$\F_{16^{13}}$ over the finite field $\F_{16}$}
{\small
\begin{verbatim}
 // Construction of the function field
 n:=13; g:=2; q:=16; F16<a>:=GF(16);
 Kx<x> := FunctionField(F16); 
 R<x>:=PolynomialRing(F16);
 Kxy<y> := PolynomialRing(Kx);
 f:=y^2 + y + x^5;
 F<c> := FunctionField(f);
 // Construction of the place Q and divisor D
 q := x^13+a^6*x^12+a^5*x^11+a^11*x^10+x^9+a^12*x^8+a^7*x^7+
 a^7*x^5+a^2*x^4+a^11*x^3+a^8*x^2+a^6*x+a^14;
 DD:=x^14+a^9*x^13+a^6*x^12+a^7*x^11+a^11*x^10+a^12*x^9+ 
 a^10*x^8+a^6*x^7+a^7*x^6+a^10*x^5+a^14*x^4+x^3+x^2+a^3*x+a; 
 P:=Decomposition(F,Zeros(Kx!q)[1]); Q:=P[1];
 D:=Decomposition(F,Zeros(Kx!DD)[1])[1]; D:=1D;
 IsSpecial(1D); // false
 Dimension(Q-D); //  0
 IsSpecial(D-Q); // false
 // Construction of the residue class field 
 // and the degree one places
 K<b>:=ResidueClassField(Q);                                                                                         
 LP:=Places(F,1);  

 // Construction of the Riemann space
 LD, h :=RiemannRochSpace(D);
 L2D, h2 :=RiemannRochSpace(2D);
 BaseLD:=[(h(v))@@h2 : v in Basis(LD)];
 Base:=ExtendBasis(BaseLD,L2D);
 L2D:=[]; 
 for i in[1..2*n+g-1] do 
   L2D:=Append(L2D, h2(Base[i])); 
 end for;
 ML2D:=Matrix(2*n+g-1,1,L2D);
 // Construction of E : E=Evalf(Q)
 L:=[]; 
 for i in [1..n] do 
   L:=Append(L,ElementToSequence(Evaluate(L2D[i],Q)));
 end for;
 E:=Transpose(Matrix(L));
 // we use a normal basis 
 PP:=[]; 
 for i:=0 to n-1 do 
   PP:=Append(PP, ElementToSequence(b^(q^i))); 
 end for;
 NC:=Transpose(Matrix(F16,n,n,PP)); NCI:=NC^-1;
 //  X.M=fx   
 M:=Matrix(F,E^-1*NC);
 Ev:=Matrix(1,n,[L2D[i] : i in [1..n]]);
 // Construction of L(2D) with the required properties
 EL2D:=Matrix(2*n+g-1,1,
       [Evaluate(L2D[i],Q) : i in [1..2*n+g-1]]);
 BEL2D:=Matrix(F16,2*n+g-1,n,
        [ElementToSequence(EL2D[i][1]) : i in [1..2*n+g-1]]);
 MM:=Parent(ZeroMatrix(F,n+g-1,2n+g-1))!
           Matrix(Basis(NullSpace(BEL2D)))*ML2D;   
 for i in [1..n] do 
   L2D[i]:=Transpose(EvM)[i,1]; 
 end for; 
 // rows of ML2D form a basis of L(2D)                  
 ML2D:=Matrix(2*n+g-1,1,
       [L2D[i] : i in [1..n]] cat [MM[i,1] : i in [1..n+g-1]]);
 // we can check that evaluation in 
 // Q of the last n+g-1 gives 0 :
 // [Evaluate(ML2D[i,1],Q) : i in [1..2*n+g-1]];

 // Construction of T and T^-1
 ST:=[]; 
 for j:=1 to 2*n+g-1 do 
   for i:=1 to 2*n+g-1 do 
     ST:=Append(ST, Evaluate(ML2D[i,1], LP[j]));
   end for; 
 end for;
 T:=Matrix(2*n+g-1, ST); 
 TI:=T^-1;
 // Construction of P
 P:=VerticalJoin(HorizontalJoin(ScalarMatrix(F16,n,1),
                         ZeroMatrix(F16,n,n+g-1)),
 ZeroMatrix(F16,n+g-1,2*n+g-1));
 // Matrix T1
 T1:=T*P*TI;
 // X and Y are the elements to multiply
 // represented in a normal basis
 X:=Matrix(F16,13,1,[a,1,0,0,0,0,0,0,0,0,0,0,0]); // example
 Y:=Matrix(F16,13,1,[1,a,a,0,0,0,0,0,0,0,0,0,0]); // example
 fx:=VerticalJoin(X,ZeroMatrix(F16,n+g-1,1));
 fy:=VerticalJoin(Y,ZeroMatrix(F16,n+g-1,1));
 //   u = T(fx)T(fy) 
 u:= Matrix(2*n+g-1,1,
          [(Tfx)[i][1]*(Tfy)[i][1] : i in [1..2*n+g-1]]);
 //  fz = MM(P*TI*u) 
 fz:=Matrix(n,1,[(P*TI*u)[i][1] : i in [1..n]]);
 // fz gives X*Y in the normal basis
\end{verbatim}
}


\end{document}